\documentclass[conference]{IEEEtran}
\usepackage{amsmath}
\usepackage{graphicx}
\usepackage{algorithm}
\usepackage{algorithmicx}
\usepackage{algpseudocode}
\usepackage{float}
\usepackage{amsfonts,amssymb}
\usepackage{enumerate,ifthen}
\usepackage{mathrsfs}
\usepackage{amsthm}
\usepackage{wrapfig}
\newtheorem{theorem}{Theorem}
\def\BibTeX{{\rm B\kern-.05em{\sc i\kern-.025em b}\kern-.08em
    T\kern-.1667em\lower.7ex\hbox{E}\kern-.125emX}}

\newcommand{\vSpacing}{\vspace*{0.1cm}}
\newcommand{\subHeadingS}[1]{ %starting subheading directly under \section{}
	\noindent \textbf{#1} %\
}
\newcommand{\subHeading}[1]{
	\vSpacing
	\noindent \textbf{#1}%\
}

\graphicspath{{Figures/}}
\graphicspath{{Figures/}{/}}
\title{Unified Scheduling for Predictable Communication Reliability in Cellular Networks with D2D Links}
\author{\IEEEauthorblockN{Yuwei Xie}
        \IEEEauthorblockA{Elec$.$ \& Comp$.$ Eng$.$ Department \\ %of Electrical \\and Computer Engineering\\
        Iowa State University\\
         yuweix@iastate.edu}
        \and
        \IEEEauthorblockN{Hongwei Zhang}
        \IEEEauthorblockA{Elec$.$ \& Comp$.$ Eng$.$ Department \\% of Electrical \\and Computer Engineering\\
        Iowa State University\\
        hongwei@iastate.edu}
        \and
        \IEEEauthorblockN{Pengfei Ren}
        \IEEEauthorblockA{Computer Science Department \\ %of Computer Science\\
        Wayne State University\\
         pengfei@wayne.edu}}

\begin{document}
\bibliographystyle{IEEEtran}
\maketitle
\begin{abstract}
Cellular networks with D2D links are increasingly being explored for mission-critical applications (e.g., real-time control and AR/VR) which require predictable communication reliability. Thus it is critical to control interference among concurrent transmissions in a predictable manner to ensure the required communication reliability. To this end, we propose a Unified Cellular Scheduling (UCS) framework that, based on the Physical-Ratio-K (PRK) interference model, schedules uplink, downlink, and D2D transmissions in a unified manner to ensure predictable communication reliability while maximizing channel spatial reuse. UCS also provides a simple, effective approach to mode selection that maximizes the communication capacity for each involved communication pair. UCS effectively uses multiple channels for high throughput as well as resilience to channel fading and external interference. Leveraging the availability of base stations (BSes) as well as high-speed, out-of-band connectivity between BSes, UCS effectively orchestrates the functionalities of BSes and user equipment (UE) for light-weight control signaling and ease of incremental deployment and integration with existing cellular standards. We have implemented UCS using the open-source, standards-compliant cellular networking platform OpenAirInterface. We have validated the OpenAirInterface implementation using USRP B210 software-defined radios and lab deployment. We have also evaluated UCS through high-fidelity, at-scale simulation studies; we observe that UCS ensures predictable communication reliability while achieving a higher channel spatial reuse rate than existing mechanisms, and that the distributed UCS framework enables a channel spatial reuse rate statistically equal to that in the state-of-the-art centralized scheduling algorithm iOrder.
\end{abstract}

\section{Introduction} \label{sec:introduction}
LTE-Advanced Pro and 5G cellular networks with device-to-device (D2D) communications are increasingly being explored for mission-critical applications such as real-time control and AR/VR \cite{Erik-LTE-book:2016,5G-MTC:overview,5G-D2D:overview}. For these applications, predictable communication reliability is not only important by itself, it is also the basis of real-time communication since unpredictable communication reliability will make it difficult to ensure timely delivery of messages \cite{PRKS, PRK}. Controlling communication reliability in a predictable manner is also a basis for controlling the inherent trade-off between communication reliability, delay, and throughput, which is important for system-level optimization \cite{PRKS, control-comm-netRT}. Cellular communication, however, is subject to complex dynamics and uncertainties, and interference among concurrent transmissions is a major source of uncertainty \cite{PRKS, PRK}. For predictable communication reliability in mission-critical cellular networks, it is critical to schedule concurrent transmissions so that interference among them is controlled in a predictable manner.

\subHeading{Related work}
For controlling interference in cellular networks with D2D links, power control, channel assignment, and scheduling have been considered in existing studies \cite{Hyunkee:capacity-d2d, Qiaoyang:resource-d2d, Jeff:D2D-resource-allication-distributed, Xuemin:resource-control-d2d, Vincent:dynamic-power, Vincent:delay-aware-control-survey, D2D:Verenzuela2017, D2D:Lv2017, doppler2009device, peng2009interference, min2011capacity, yu2009performance, janis2009interference, feng2016qos}.   The existing cellular technology LTE has also defined the High Interference Indicator (HII) and Overload Indicator (OI) for uplinks as well as Relative Narrowband Transmit Power (RNTP) for downlinks in order for a cell to inform neighboring cells of the Resource Blocks (RBs) that are susceptible to interference \cite{4GLTE}. 
For more precise interference control in coordinated multi-point (CoMP) transmission and reception, LTE has also defined coordinated scheduling mechanisms such as dynamic point blanking which dynamically prevent transmission at certain time-frequency resource.  
The existing mechanisms, however, do not ensure predictable interference control and communication reliability due to the following reasons: considering single-cell settings only without addressing inter-cell interference \cite{Hyunkee:capacity-d2d, Qiaoyang:resource-d2d, Jeff:D2D-resource-allication-distributed, Xuemin:resource-control-d2d}, using inaccurate interference models \cite{Vincent:dynamic-power, Vincent:delay-aware-control-survey}, assuming uniform wireless channel fading across networks which is unrealistic in practice \cite{D2D:Verenzuela2017, D2D:Lv2017}, assuming exclusion-regions around receivers without mechanisms for identifying the exclusion-regions in practice \cite{Hyunkee:capacity-d2d}, not considering the interference from cellular-mode transmitters to D2D receivers \cite{doppler2009device}, not considering interference between D2D links \cite{peng2009interference}, not considering inter-cell interference among cellular and D2D links \cite{min2011capacity,yu2009performance,min2011reliability}, and/or maximizing throughput without addressing reliability-throughput tradeoff \cite{Hyunkee:capacity-d2d,Qiaoyang:resource-d2d,Jeff:D2D-resource-allication-distributed}. Several pieces of work \cite{min2011capacity,yu2009performance,min2011reliability} do not consider spectrum reuse across D2D links either, which unnecessarily reduces achievable network capacity. 

Wireless networks such as those based on WirelessHART, ISA100.11p, WIA-PA, and IETF 6TiSCH \cite{IndustrialNetBook-CRC,Lu:WirelessHART-Sch} have been studied for industrial applications. Focusing on low-rate wireless networks based on IEEE 802.15.4/4e, those studies do not focus on cellular networks with D2D links, nor have they focused on distributed scheduling with predictable interference control and maximum channel spatial reuse \cite{PRKS}. 

\subHeading{Contributions of this work}
Towards predictable communication reliability in industrial cellular networks with D2D links, we propose a Unified Cellular Scheduling (UCS) framework and we make the following contributions:
\begin{itemize}
\item Based on the Physical-Ratio-K (PRK) interference model which is suitable  for  developing  field-deployable  distributed  scheduling algorithms \cite{PRK}, our UCS framework schedules uplink, downlink, and  D2D  transmissions  in  a  unified  manner  to  ensure  predictable communication reliability while maximizing channel spatial reuse and allocating  communication  resources  to  uplink,  downlink,  and  D2D transmissions on a need basis. UCS also provides a simple, effective approach to mode selection that maximizes the communication capacity for each involved communication pair. 
\item Extending the distributed scheduling protocol PRKS \cite{PRKS} to multi-channel settings, UCS effectively uses multiple communication channels  for  high  throughput  as  well  as for resilience to channel fading and external interference. 
\item To leverage the computational power of base stations (BSes) as well as high-speed, out-of-band connectivity between BSes, UCS places the scheduling decisions at BSes and having UEs share their local state information  with  corresponding  BSes  at  relatively  low-frequencies. This  BS-UE  functional  orchestration  mechanism  enables light-weight control signaling, and it facilitates incremental deployment of UCS as well as technology evolution.
\item We  have  implemented  UCS  using  the  open-source,  standards-compliant  cellular  networking  platform  OpenAirInterface.  We have validated the OpenAirInterface implementation of UCS using USRP B210 software-defined radios and a small-scale deployment. We have also studied the behavior of UCS using at-scale, high-fidelity simulation. We have observed that, unlike existing mechanisms which cannot enable predictable communication reliability, UCS ensures predictable communication reliability while achieving higher channel spatial reuse rate. We have also observed that the distributed UCS scheduling framework enables a channel spatial reuse rate statistically equal to that in the state-of-the-art centralized scheduling algorithm iOrder \cite{che2014case}. 
\end{itemize}

The rest of the paper is organized as follows. Section~\ref{sec:preliminaries} presents the system model, problem specification, PRK interference model, and PRK-based scheduling protocol PRKS. Section~\ref{sec:UCS} presents our Unified Cellular Scheduling (UCS) framework. We evaluate UCS in Section~\ref{sec:simulationEval}, and we summarize our concluding remarks in Section~\ref{sec:conclusion}.

\section{Preliminaries} \label{sec:preliminaries}
\subsection{System Model and Problem Specification}
We consider cellular networks of multiple cells where each cell has a Base Station (BS) and a number of user equipment (UEs). Each cell has a set of uplinks (i.e., transmissions from UEs to the BS) and downlinks (i.e., transmissions from the BS to UEs). A UE may also transmit data to another UE in the network. The transmission from one UE to another can be in cellular mode (i.e., an uplink transmission followed by a downlink transmission) or in D2D mode where the transmitter UE sends data directly to the receiver UE without using any BS in data delivery. If a UE sends data directly to another UE, we regard the communication link as a D2D link. In line with the current wireless systems, e.g., LTE-type systems, the basic resource allocation unit is Resource Block (BS), which consists of 12 consecutive subcarriers in the frequency domain and one 0.5ms time slot in the time domain, with each subcarrier occupying a 15KHz spectrum and the central frequencies of two consecutive subcarriers separated by 15KHz. For convenience of exposition, we regard the 12 consecutive subcarriers of a RB as one carrier. According to the LTE standard, each cell may use multiple component carriers, with each component carrier consisting of 6 - 100 carriers (i.e., with bandwidth ranging from 1.4MHz to 20MHz) \cite{4GLTE}. For reducing scheduling overhead, the LTE standard also groups a certain number carriers into a carrier group, and the specific grouping methods depend on the bandwidth of a component carrier.

The uplinks, downlinks, and D2D links of a cellular network share the wireless spectrum available. The objectives of this work are to develop 1) an algorithm that, for each UE-to-UE communication pair, decides whether the communication shall be in cellular mode or D2D mode for maximum communication throughput while satisfying the required communication reliability, and 2) an algorithm that, given a time slot and the set of uplinks, downlinks, and D2D links (if any), schedules a maximal subset of the links to transmit at the time slot so that the required communication reliability is guaranteed.

As a first-step towards field-deployable solutions that ensure predictable communication reliability in cellular networks with D2D links, we assume that UEs are static or move slowly such that the average channel gain and the average background noise power tend to be stable at timescales of seconds, minutes or even longer (e.g., in industrial network settings of slow motion) \cite{5G-MTC:overview,PRKS}. Focusing on the problem of transmission scheduling, we assume that transmission power for each node (i.e., BS or UE) is fixed, even though different nodes may use different transmission powers. For highly mobile networks (e.g., those with vehicles) and transmission power control, techniques such as those by Li et al$.$ \cite{CPS-IOTDI18} and Wang et al$.$ \cite{Zhang:pktReliability} may be applied, but detailed study is beyond the scope of this work.

\subsection{Interference Model} \label{subsec:interferenceModel}

\begin{wrapfigure}{r}{3.5cm}
	\includegraphics[width=.99\linewidth]{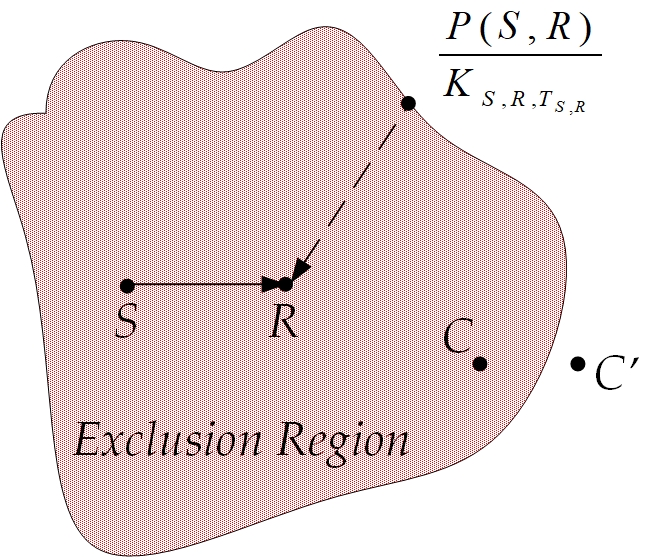}
	\caption{PRK interference model} \label{fig:PRK-model}
	\vspace*{-0.1cm}
\end{wrapfigure}
For predictable interference control in transmission scheduling, we adopt the Physical-Ratio-K (PRK) interference model \cite{4GLTE} in our study. The PRK model integrates the protocol 
model's locality with the physical model's high-fidelity, and 
it is suitable for designing distributed scheduling protocols that ensure predictable interference control in the presence of dynamics and uncertainties. 
As shown in Figure~\ref{fig:PRK-model},
in the PRK model,  
a node $C'$ is regarded as not interfering and thus can transmit concurrently with the transmission form another node $S$ to its receiver $R$ if and only if $P(C',R)<\frac{P(S,R)}{K_{S,R,T_{S,R}}}$, 
\iffalse 
the following holds:
\begin{align}
P(C',R)<\frac{P(S,R)}{K_{S,R,T_{S,R}}}
\end{align}
\fi
where $P(C',R)$ and $P(S,R)$ is the average strength of signals reaching R from $C'$ and $S$ respectively, $K_{S,R,T_{S,R}}$ is the
minimum real number chosen such that, in the presence of
cumulative interference from all concurrent transmitters, the
probability for $R$ to successfully receive packets from $S$ is
no less than the minimum link reliability $T_{S,R}$ required by applications.

For predictable interference control, the parameter $K_{S,R,T_{S,R}}$ of the PRK model needs to be instantiated for every link $(S,R)$ according to in-situ, potentially unpredictable network and environmental conditions. To this end, Zhang et al \cite{PRKS} have formulated the PRK model instantiation problem as a regulation control problem where the ``plant`` is the link $(S,R)$, the ``reference input`` is the required link reliability $T_{S,R}$, the ``output`` is the actual link reliability $Y_{S,R}$ from $S$ to $R$, the ``control input`` is the PRK model parameter $K_{S,R,T_{S,R}}$, and the objective of the regulation control is to adjust the control input so that the plant output is no less than the reference input \cite{PRKS}. For every link $(S,R)$, using its instantiated PRK model parameter $K_{S,R,T_{S,R}}$ and the local signal maps that contain the average signal power attenuation between $S$, $R$ and every other close-by node $C$ that may interfere with the transmission from $S$ to $R$, link $(S,R)$ and every close-by node $C$ become aware of their mutual interference relations. Based on nodes/links’ mutual interference relations, non-interfering transmissions can be scheduled to ensure the required communication reliability across individual links. 

PRK-based scheduling has been shown to enable predictable interference control in single-channel ad hoc networks \cite{PRKS,CPS-IOTDI18}. In this work, we will verify the suitability and address the challenges of applying the PRK model to scheduling in multi-channel cellular networks with D2D links.

\section{Unified Cellular Scheduling Framework}\label{sec:UCS}

\subsection{Overview} \label{subsec:UCS-overview} 

For scheduling with predictable communication reliability in cellular networks, a fundamental task is to identify the interference relations between uplinks, downlinks, and D2D links (if any). %which may well be distributed in multiples cells. 
Given that the PRK interference model is a high-fidelity model specifically designed for distributed protocol design in dynamic, uncertain network settings \cite{PRK} and considering the demonstrated predictable interference control in PRK-based scheduling for single-channel ad hoc networks \cite{PRKS}, we adopt the PRK interference model in our design. The PRK model is a generic model, and it is applicable to communication links of different technologies. In particular, the impact of different communication technologies (e.g., modulation and coding schemes, multi-antenna systems) is captured by the relation between the packet-delivery-reliability (PDR) and signal-to-interference-plus-noise-ratio (SINR) for a link, and the PDR-SINR relation is used to instantiate the PRK model for each link \cite{PRKS}. Therefore, with the PRK model instantiated for the uplinks, downlinks, and D2D links (if any) of a cellular network, the PRK model serves as a unified approach to modeling interference relations between uplinks, downlinks, and D2D links despite the differences between these links (e.g., different types of transmitter/receiver radios). Accordingly, with inter-link interference relations identified, the cellular network scheduling problem is transformed into a unified problem of identifying maximal independent sets in a conflict graph capturing the inter-link interference relations, thus enabling maximizing spectrum spatial reuse while ensuring predictable communication reliability. 

The availability of multiple communication channels (e.g., carriers) in cellular networks introduces the scalability challenge of PRK-based scheduling (e.g., in control signaling overhead), and it also provides the opportunity of channel-hopping for increased resilience against channel fading and external interference. In general, a link may maintain one PRK model parameter $K$ for a group of $n$ channels (e.g., carrier, carrier group, or component carrier), and the choice of $n$ reflects the tradeoff between control signaling overhead, data communication performance, and ease of implementation with existing LTE standard framework and thus incremental deployment. We will analyze the tradeoff in Section~\ref{PRK adaptation}, and we will present a PRK-based multi-channel scheduling algorithm that leverages channel hopping to balance communication load across multiple channels and to increase resilience against channel fading and external interference. 

The PRK model unifies the scheduling of uplinks, downlinks, and D2D links in cellular networks, thus PRK-based cellular network scheduling bears similarity to that in ad hoc networks and provides a unified framework for reasoning about interference-control-oriented wireless network scheduling. In addition, the availability of BSes and high-speed, out-of-band interconnections between BSes (e.g., wired optical networks) in cellular networks provide unique opportunities of orchestrating the functionalities of BSes and UEs in ways to reduce control signaling overhead and to facilitate incremental deployment and technology evolution. 

Figure~\ref{fig:UCS-arch} shows the architecture of the Unified Cellular Scheduling (UCS) framework. 
\begin{figure}[!htbp]
	\vspace*{-0.05in}
\centering
\includegraphics[width=2.5in]{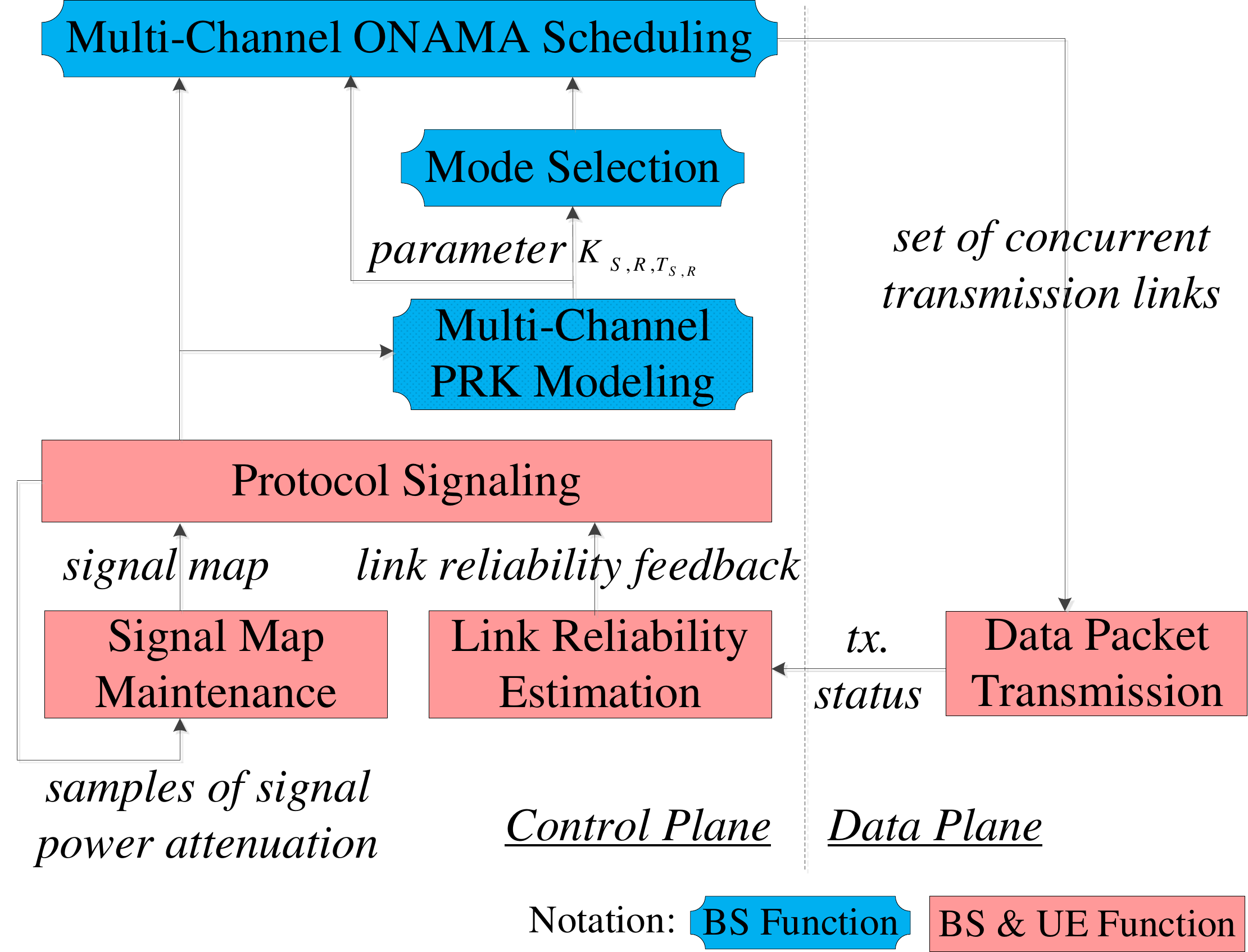}
\caption{UCS architecture} \label{fig:UCS-arch}
\vspace*{-0.1in}
\end{figure}  
In the architecture, based on the status (i.e., success or failure) of uplink transmissions as well as downlink and D2D link transmissions, the BSes and UEs estimate the uplink communication reliability as well as the downlink and D2D link communication reliability respectively. Then, through collaborative control signaling, close-by nodes (i.e., BSes and UEs) estimate the average channel gain between themselves and use local signal maps to record such estimates. To facilitate incremental deployment and technology evolution, we keep the functionality of UEs to the minimum, and UEs share their communication reliability estimates and local signal maps with their respective BSes. This way, the BSes collectively have all the information needed to estimate the PRK model parameters of uplinks, downlinks, and D2D links and then make decisions on mode selection and transmission scheduling. Then, each BS can inform the corresponding UEs of the transmission modes and schedules using existing LTE signaling mechanisms or their simple extensions. 
	In what follows, we elaborate on the design and key components of the UCS framework.

\subsection{Multi-Channel PRK Modeling} \label{PRK adaptation}

PRK-based scheduling guarantees predictable communication reliability by maintaining one PRK model parameter $K$ for each channel in ad hoc networks, but it is neither practical nor necessary in cellular networks with D2D links. In the LTE-type systems, each component carriers may contain up to 100 carriers and each cell may have multiple component carriers, which means there are hundreds of channels for each link. For maintaining one $K$ for each channel, a node needs to transmit the average channel gain between itself and neighboring nodes for each channel to the BS. In this situation, the number of pieces of control information can reach tens of thousands in each feedback period. The excessive control information will occupy resources for data transmission, resulting in scheduling delay and reducing overall system throughput. Consequently, maintaining one $K$ for each channel is not practical in the LTE-type cellular networks.

To reduce control signaling overhead, LTE defines the carrier group as the scheduling unit depending on the bandwidth of the component carrier, which means that BSes do not need to know the channel gain for each carrier. Maintaining one $K$ for each channel in cellular networks may not be able to further improve the accuracy of interference control, and it will also introduce control signaling overhead.

Maintaining one $K$ for too many carriers, however, will increase the number of control steps needed for the system to converge. When a new link requests transmissions or when environment changes, inaccurate $K$ may incorrectly add nodes to an exclusion region or delete nodes from it due to the differences of those carriers in the frequency domain, which means the scheduler needs more steps to make the value of $K$ converge. The average link communication reliability will be affected by the increased convergence time. Consequently, the number of carriers sharing a common PRK model parameter $K$ should be properly selected to guarantee control accuracy.

UCS maintains one $K$ for a certain number of adjacent carriers on one component carrier to guarantee predictable communication reliability. As we will show in Section~\ref{subsec:multi-channel-ONAMA}, the availability of multiple carriers for a shared $K$ enables UCS to the channel hopping to increase resilience against channel fading and external interference for improved communication reliability. Experiments have shown that a suitable number of carriers sharing one $K$ can not only guarantee predictable communication reliability but reduce control signaling overhead greatly. For example, consider the scenario where $20$ UEs share a component carrier with 100 carriers. If it maintains one $K$ for each carrier, each UE needs to transmit  $19\times 100=1,900$ pieces of information about average channel gains, and the whole system needs to transmit $1,900\times 20=38,000$ pieces of average channel gain information during each feedback period. If it maintains one $K$ for every $25$ adjacent carriers while also guaranteeing predictable communication reliability, each UE just needs to transmit $19\times 4=76$ average channel gain information and the whole system needs to transmit $76\times 20=1,520$ pieces of average channel gain information during each feedback period. The control information will be reduced by more than $90\%$.

For each link $(S, R)$, the PRK model parameter $K$ is initialized such that the initial exclusion region around the receiver $R$ includes every strong interferer whose concurrent transmission alone in the same carrier as that of $(S, R)$ can make the communication reliability along $(S, R)$ drop below the required reliability $T_{S,R}$. 
	After its initialization, the PRK model parameter is adapted according to in-situ network and environmental conditions to ensure the required communication reliability, and the adaptation uses the regulation feedback control mechanism of Zhang et al$.$ \cite{PRKS} which we have discussed in Section~\ref{subsec:interferenceModel}. The method of Zhang et al$.$ \cite{PRKS} was proposed for single-carrier networks. In cellular networks of $N$ wireless carriers and using the scheduling algorithm to be presented in Section~\ref{subsec:multi-channel-ONAMA} shortly, a link uses a specific carrier at a time slot in probability $\frac{1}{N}$. Assuming that the average strength of the interference signal from an interferer $C$ to the receiver $R$ is $P(C, R)$ when both $C$ and $R$ are in the same communication carrier, the expected interference from $C$ to $R$ shall be computed as $\frac{P(C, R)}{N}$ (instead of as $P(C, R)$ in the original feedback control mechanism \cite{PRKS}) since $C$ may not use the same carrier as $R$.

\subsection{Mode Selection}
The aim of mode selection for UE-to-UE communication pairs is to accommodate as many concurrently-transmitting links as possible for a given communication reliability requirement. Consequently, a UE needs to choose a transmission mode that allows more links to transmit concurrently with it. To this end, the mode selection question becomes a comparison of the number of concurrent links between D2D transmission mode and cellular transmission mode. The PRK model provides a good basis for addressing this issue because the value of $K$ exactly determines which links can or cannot transmit concurrently.

From the PRK interference model, each link maintains a $K$ to satisfy its reliability requirement, and thus each $K$ defines an exclusion region $\mathbb{E}$. For each link $(S,R)$, node $C$ is in the exclusion region of the receiver $R$ and thus shall not transmit concurrently with the transmission from $S$ to $R$ if and only if $P(C,R)\ge \frac{P(S,R)}{K_{S,R,T_{S,R}}}$ holds. For the D2D transmission mode, only one transmission link is needed, from the transmitter to the receiver directly. For the cellular transmission mode, two transmission links are required - an uplink followed by a downlink. This means that a UE needs to compare the number of interference links of the D2D link with the total number of interference links of both the uplink and downlink to make the decision. Consequently, the mode selection algorithm is based on the number of nodes in the exclusion region defined by $K$. When a pair of UEs request transmissions, the BS calculates the exclusion region for the D2D link ($\mathbb{E}_{D2D}$), uplink ($\mathbb{E}_{Up}$) and downlink ($\mathbb{E}_{Down}$) respectively according to their PRK model parameters. If $\vert\mathbb{E}_{D2D}\vert>\vert\mathbb{E}_{Up}\vert+ \vert\mathbb{E}_{Down}\vert$, the UEs shall communicate through the BS in the cellular mode. If $\vert\mathbb{E}_{D2D}\vert<\vert\mathbb{E}_{Up}\vert+ \vert\mathbb{E}_{Down}\vert$, the UEs shall communicate directly in the D2D mode.

To realize the above design in practice, we shall consider the fact that the values of $\vert\mathbb{E}_{D2D}\vert$ and $ \vert\mathbb{E}_{Up}\vert + \vert\mathbb{E}_{Down}\vert$ for a UE-to-UE communication pair have to be learned by the UEs, the UEs cannot learn about $\vert\mathbb{E}_{D2D}\vert$(or $ \vert\mathbb{E}_{Up}\vert + \vert\mathbb{E}_{Down}\vert$) unless they communicate in the D2D (or cellular) mode, and $\vert\mathbb{E}_{D2D}\vert$ as well as  $ \vert\mathbb{E}_{Up}\vert + \vert\mathbb{E}_{Down}\vert$ are potentially time-varying. In particular, the mode selection problem considering these real-world challenges can be modeled as a restless multi-armed bandit (MAB) problem \cite{tekin2015online}. A MAB problem can be seen as a set of real distributions $B=\{R_1,...,R_K\}$, with each distribution $R_k (1 \le k \le K)$ being associated with the rewards delivered by the $k$-th arm (i.e., decision options) and having a mean value of $\mu_k$. A gambler iteratively plays one arm at a time and collects the associated reward. The objective is to maximize the sum of the collected rewards. For a given play strategy, the realized regret $\rho$ after $T$ plays is defined as follows: 
\[
R_{real}(T)=T\mu^*-\sum_{t=1}^T(\mu_{\alpha_t}-\beta_t c_{\alpha_t})
\]
where $\mu^* = \max_{k=1}^K \mu_k$, $\alpha_t$ denotes the arm selected at time $t$, $c_{\alpha_t}$ is the cost of observing the reward of arm $\alpha_t$, $\beta_t$ is equal to 1 if the gambler observes the reward of $\alpha_t$ at time $t$ and $0$ otherwise. 

In the mode selection problem, we can treat the rewards of the D2D mode and cellular mode as $-\vert\mathbb{E}_{D2D}\vert$ and $ -\vert\mathbb{E}_{Up}\vert - \vert\mathbb{E}_{Down}\vert$ respectively, and then we can apply the History-Dependent Sequencing of Exploration and Exploitation (HD-SEE) \cite{tekin2015online} algorithm in mode selection. 
In particular, let $S_{k,t}$ be the number of reward observations for arm $k$ by time $t$ and $N_{k,t}$ be the number of times arm $k$ has been selected by time $t$. The realized regret can be rewritten as follows:
\[%equation}
R_{real}(T)=\sum_{k=1}^K(\vartriangle_kN_{k,T}+c_kS_{k,T})
\]
where $\triangle_k=\mu^*-\mu_k$ is called the suboptimality gap of arm $k$.
At each time $t$, the algorithm starts by calculating the estimated optimal arm: 
\[
\hat{k}_t^*=\mathop{\arg\max}_{k\in\mathcal{K}}\hat{\mu}_{k,t},
\]
where $\hat{\mu}_{k,t}$ is the estimated mean reward of arm $k$ by time $t$. 
Denoting the set of available arms as $\mathcal{K}$, then, for each arm $k\in \mathcal{K}-k_t^*$, it calculates the estimated suboptimality gap
\[
\hat{\vartriangle}_{k,t}=\hat{\mu}_{k_t^*,t}-\hat{\mu}_{k,t}.
\]
For each arm $k\in \mathcal{K}$, a control number $D_{k,t}$ is calculated based on the estimated suboptimality gap and the number of times that arm has been explored. For $k\in \mathcal{K}-k_t^*$ the control number is given as
\[
D_{k,t}=\frac{L_2\log(tK/\delta)}{J_{k,t}^2}
\]
where
\[
J_{k,t}=max\left\{0,\hat{\triangle}_{k,t}-2\sqrt{\frac{L_1\log(tK/\delta)}{\min(S_{k,t},S_{k_t^*,t})}}\right\}.
\]
Here $L_1>0$ and $L_2>0$ are constants. The control number for the estimated optimal arm $k_t^*$ is calculated as
\[
D_{k_t^*,t}=\frac{L_2\log(tK/\delta)}{\min_{k\in \mathcal{K}-k_t^*}J_{k,t}^2}
\]
For the selection process, if $S_{k,t}\geqslant D_{k,t}$ for all $k\in \mathcal{K}$, the algorithm chooses the arm $\hat{k}_t^*$, if not, the algorithm randomly chooses an arm for which $S_{k,t}<D_{k,t}$.In the mode selection problem, each BS only has two arms. Therefore, if $S_{k,t}\geqslant D_{k,t}$ holds for the suboptimal arm, the BS chooses the optimal arm; otherwise, the BS chooses the suboptimal arm.

In the HD-SEE algorithm, the number of times a suboptimal arm is chosen depends on the suboptimality gap of the arm. An arm with a larger suboptimality gap will be explored (i.e., chosen) fewer number of times compared to an arm with a smaller suboptimality gap. The regret in the HD-SEE algorithm is logarithmic in time, which is the best possible \cite{tekin2015online}.

\subsection{BS and UE Functional Orchestration}
\label{subsec:functionalOrchestration}
The main difference between cellular networks and ad hoc networks is the availability of BSes in cellular networks. Cellular networks also provide high-speed, out-of-band interconnections between BSes to exchange control signaling information (e.g., those needed for transmission scheduling). The availability of BSes and the high-speed, out-of-band interconnections in between provides the opportunity of having each BS collect information about network state information in its cell (e.g., local signal maps containing wireless channel gains and communication reliability across different links) through existing LTE BS-UE control signaling mechanisms, having BSes coordinate with one another in deciding network-wide transmission schedules, and then having each BS inform UEs in its cell of their transmission modes and schedules. This approach of placing core intelligence (i.e., decision making logic) at BSes and keeping UE’s functionality to the minimum of estimating communication reliability and maintaining local signal maps helps facilitate incremental deployment and technology evolution; this is because the number of BSes tend to be much less than that of UEs.
% and the BSes are usually managed by individual service providers whose number is much less than the number of owners of UEs. 

UCS also utilizes the inter-cell interference coordination (ICIC) mechanism of cellular networks to reduce control signaling overhead. The ICIC mechanism transmits some indicator message using the X2 interface to help with the scheduling process of neighboring BSes. The transmission along the X2 interface usually adopts high-speed wired networks such as optical fiber networks and does not consume wireless spectrum. Therefore, X2 interface transmission can be used to further reduce control signaling overhead in wireless channels. From the PRK interference model, to calculate the exclusion region for each link, the BS needs to know the receiver's local signal map and the value of $Ks$ of the close-by links. However, some close-by links are in the neighboring cells and thus require inter-cell coordination. In UCS, the values of the PRK model parameter $K$ and local signal maps are treated as the data part of the ICIC-related messages, and they are transmitted to the neighboring cell using the X2 interface to avoid using wireless resources. By the ICIC mechanism, transmission of inter-cell coordination control information no longer requires wireless transmissions.

\subsection{Multi-Channel ONAMA Scheduling} \label{subsec:multi-channel-ONAMA}

Based on the interference relations between the uplinks, downlinks, and D2D links (if any as a result of mode selection) as identified by the PRK interference model, data transmissions along all the links can be scheduled in a unified manner to fully utilize the available wireless communication carriers. In particular, the objective of the unified scheduler is to schedule data transmissions so that a maximal set of non-interfering links are scheduled to transmit at each carrier and each time slot and that a link is scheduled to transmit only if there exists at least one data packet queued for transmission at the beginning of a time slot. Unlike existing cellular network scheduling algorithms which are based on a limited set of preconfigured frequency-division-duplexing (FDD) or time-division-duplexing (TDD) transmission patterns, the unified scheduler is adaptive to application traffic demand, and it ensures predictable communication reliability by respecting the interference relations as identified by the PRK model.

More specifically, the unified scheduler is based on the ONAMA TDMA scheduling algorithm \cite{liu2015maximal}. ONAMA schedules a maximal set of non-interfering links to transmit at each time slot, but it is designed for single-carrier wireless networks, and it is not adaptive to traffic demand. In this study, we extend the ONAMA algorithm to consider the specific cellular network properties such as the availability of multiple carriers and base stations (BSes) as well as the traffic demand across individual links. Based on the BS and UE functional orchestration mechanism presented in Section~\ref{subsec:functionalOrchestration}, each BS $I$ knows the set of transmitters in its cell and the associated links $L_I$ (whose receivers may be in a neighboring cell), and, for each such link $i$, the BS also knows the set of links $M_i$ that interfere with link $i$. For each time slot $t$, the BS can also get the traffic demand $d_i$ (i.e., number of data packets queued for transmission) for each link $i \in L_I$, and the BS knows the set of available carriers $RB$. Then, the multi-channel transmission schedule for each time slot is identified by the BSes in a distributed manner as follows: %,  in which the state of a link can be UNDECIDED, ACTIVE, or INACTIVE in each carrier: 
\begin{enumerate}[1)]
\item Each BS $I$ initializes the state of each link $i \in L_I$ as  UNDECIDED for each carrier $rb \in RB$, and $I$ also sets the state of every link in $M_i$ as UNDECIDED for each carrier; 
\item For each link $i \in L_I$, the BS $I$ computes a priority for each link $k \in M_i\cup i$ and each carrier $rb \in RB$ for $d_k$ times:
\[%\]\begin{equation}
Prio.k.rb.d = Hash(k\oplus d\oplus t\oplus rb)\oplus k\oplus d, 1\leq d\leq d_k, 
\]%\end{equation}
where $Hash(x)$ is a message digest generator that returns a random integer by hashing $x$. Note the fourth and fifth XOR operator $\oplus$ are  necessary for guaranteeing that all links' priorities are distinct even when $Hash()$ returns the same number on different inputs. This also means that a link with demand $d_k$ will have $d_k$ different priorities.

\item For each link $i \in L_I$, the BS $I$ computes a priority for each link $k \in M_i\cup i$ and each carrier $rb \in RB$:
\[%\]\begin{equation}
Prio.k.rb =  Max_{d=1}^{d_k} Prio.k.rb.d
\]%\end{equation}
That is, each link $k \in M_i\cup i$ maintains a specific priority for each available carrier $rb$.

\item  For each link $i \in L_I$, the BS $I$ iterates the following steps until the state of $i$ in each carrier is either ACTIVE or INACTIVE: A) for the carriers in which the state of $i$ is UNDECIDED, $I$ tries to assign a different state to $i$ in the increasing order of the IDs of the carriers; for a given carrier $rb$, if $i$'s priority is higher than that of every other ACTIVE and UNDECIDED member in $M_i$, $i$'s state is set as ACTIVE in carrier $rb$, and its traffic demand $d_i$ is reduced by one; conversely, if any ACTIVE member of $M_i\cup i$ has a higher priority than $i$, the state of $i$ in carrier $rb$ is set as INACTIVE; if the traffic demand of $i$ becomes zero, $i$'s state is set as INACTIVE for each carrier in which its state is UNDECIDED; B) the BS $I$ shared the state of $i$ with other BSes whose cells may have links that interfere with $i$.

\item  If the state of a link $i$ is ACTIVE for carrier $rb$ at time slot $t$, link $i$ can transmit a data packet at carrier $rb$ and time slot $t$.
\end{enumerate}
The details of the above multi-channel ONAMA scheduling algorithm for time slot $t$ are shown in Algorithm~\ref{algo:mcONAMA}.
{\small 
\begin{algorithm}[!tb]
  \caption{Multi-Channel ONAMA Scheduling at BS $I$}
  \label{algo:mcONAMA}
   \begin{algorithmic}[1]
   \item[] $M_i$: set of interfering links of a link $i \in L_I$;
   \item[] $d_k$: traffic demand of link $k \in M_i \cup i$; 
   \item[] \textbf{Perform the following actions for $\forall i \in L_I$:}
   \item state.i.rb = UNDECIDED, \ \ %rbmap.i.rb=false, 
   $\forall rb \in RB$; 
   \item[] \emph{Step 1: Priority Calculation}
   \State $Prio.k.rb.d = Hash(k\oplus d\oplus t\oplus rb)\oplus k\oplus d, \ \ \forall k \in M_i \cup i, \forall rb \in RB, \forall d\in [1,d_k]$;
   \State $Prio.k.rb =  Max_{d=1}^{d_k} Prio.k.rb.d, \forall k \in M_i \cup i, \ \ \forall rb \in RB$;
   \iffalse 
   \For{each $k \in M_i \cup i$}
      \For{each $rb \in RB$}
       \For{$d\in [1,d_k]$}
     \State $Prio.k.rb.d = Hash(k\oplus d\oplus t\oplus rb)\oplus k\oplus d$;
     \EndFor
     \State $Prio.k.rb =  Max_{d=1}^{d_k} Prio.k.rb.d$;
     \EndFor
    \EndFor
    \fi 
    
    \item[] \emph{Step 2: State Selection (i.e., Scheduling)}
    \State done = false;
    \While{done == false}
    \State done = true; 
    \For{each $rb \in RB$ in increasing order of $rb$ ID}
    \If{$d_i > 0$ \&\& state.i.rb == UNDECIDED \&\&  $Prio.i.rb > Prio.k.rb$ for each ACTIVE/UNDECIDED $k \in M_i$}
    \State state.i.rb = ACTIVE;
    \State $d_i=d_i-1$
      \If{$d_i == 0$}
      \State state.i.rb2 = INACTIVE, for each $rb2 \in$
      \State $RB$ where state.i.rb2 == UNDECIDED;
      \EndIf
    \EndIf
    \If{$Prio.i.rb < Prio.k.rb$ for any ACTIVE $k \in M_i$}
    \State state.i.rb = INACTIVE;
    \EndIf 
    \If{state.i.rb == UNDECIDED}
    \State done = false;
    \EndIf 
    \EndFor
    \State Share $state.i.rb, \forall rb \in RB$; update $state.k.rb, \forall k \in M_i, \forall rb \in RB$ based on information from other BSes;
    \EndWhile
    
    \iffalse 
    \item[] \emph{Step 3: Scheduling Decision}
    \For{each $rb \in RB$}
        \If{state.i.rb= ACTIVE}
        \State rbmap.i.rb=true;  (i.e., schedule $i$ to transmit at carrier $rb$)
        \EndIf
    \EndFor
    \fi 
    
   \end{algorithmic}
\end{algorithm}

} %small 

Similar to the original ONAMA algorithm \cite{liu2015maximal}, the above algorithm can be readily shown to converge for each time slot.  In particular, we have the following:
\begin{theorem} 
The set of all ACTIVE links at each time slot is a maximal set of non-interfering links for each carrier. 
\end{theorem}
\begin{proof}

When the iteration terminates, a link is either ACTIVE or INACTIVE in each carrier. For each INACTIVE link $i$ in carrier, there always exists an ACTIVE neighboring link in the set, whose priority is higher than that of $i$. Adding any additional link to the set of ACTIVE links for a given time slot - carrier resource block would end up with having two interfering links scheduled to transmit in the same time slot - carrier resource block , which is not allowed. Hence, The set of all ACTIVE links at a time slot $t$ is a maximal set of non-interfering links in each carrier. 

\end{proof}

It usually takes a few rounds of coordination between close-by BSes for Algorithm~\ref{algo:mcONAMA} to converge to a transmission schedule for each time slot. Given that the number of BSes is usually much less than that of UEs, the convergence is much faster than if we have every UE participate in the scheduling process. To make sure that the schedule is readily available when the time reaches a time slot $t$, the schedule for $t$ is pre-computed before time reaches $t$, and the schedules of consecutive time slots are computed in a pipeline manner as in the original ONAMA algorithm \cite{liu2015maximal}. 
	Due to randomization in the above algorithm (e.g., in computing priorities), a link may well use different carriers across different time slots, and this channel hopping behavior can help improve resilience against channel fading and external interference. Additionally, the above algorithm considers traffic demands of different links, and a link with higher demands is more likely to get the highest priority to transmit.

\section{Experimental Evaluation} \label{sec:simulationEval}
We have implemented the UCS framework in the OpenAirInterface cellular network platform \cite{OAI}. In what follows, we first evaluate the feasibility and effectiveness of UCS using a small-scale indoor deployment of software-defined-radio (SDR) implementation of UCS, and then we evaluate UCS through at-scale, high-fidelity simulation.

\subsection{SDR Prototyping and Evaluation}
\subHeading{Hardware platform}
We have implemented the UCS framework using the OpenAirInterface platform. OpenAirInterface is an open-source prototyping and experimentation platform for cellular networks, and it can be used for system simulation, system emulation, and real-world deployment and measurement \cite{OAI, virdis2015performance}. The protocol stack of OpenAirInterface is shown in Figure~\ref{oai_stack}, and we have implemented the UCS framework by modifying its MAC component.

\begin{figure}[!htbp]
	\vspace*{-0.05in}
\centering
\includegraphics[width=2.5in]{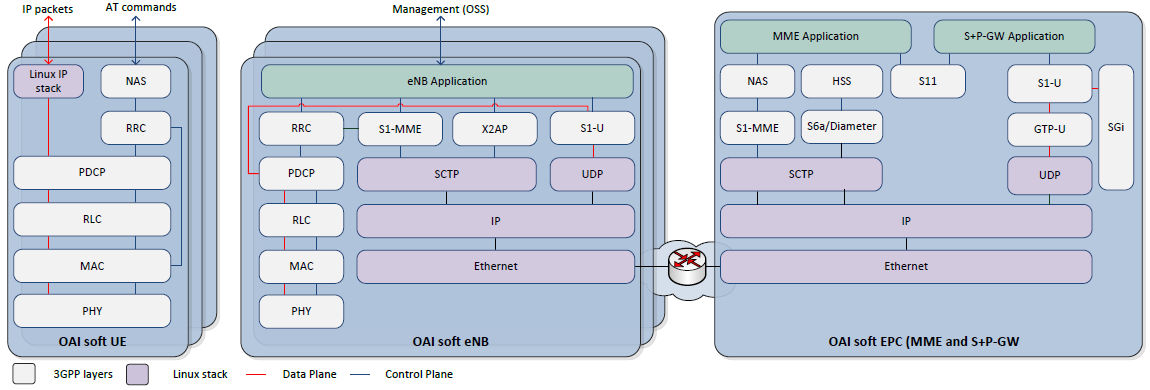}
\caption{The OpenAirInterface protocol stack}\label{oai_stack}
\vspace*{-0.1in}
\end{figure}  

To validate the feasibility of realizing the UCS framework with the LTE-compatible OpenAirInterface platform and commodity hardware, we deploy our OpenAirInterface implementation with the USRP B210 software-defined-radio (SDR) hardware in a single-cell cellular network. In the network, one Dell desktop with USRP B210 serves as the BS, and three Intel NUC mini PCs with USRP B210 serve as UEs. UE1 communicates with the BS directly, forming a cellular link. UE2 and UE3 are configured to communicate with each other, forming a UE-to-UE communication pair. The UE-to-UE communication can be in cellular mode or D2D mode, and the decision is to be made by UCS. The network is deployed inside a university building, with UEs randomly distributed such that the signal power attenuation between the UEs and BS is about 50dB. The packet-delivery-reliability (PDR) requirement for the communications is set as 90\%. 
The measurement study setup is showed in Figure~\ref{environment}.
\iffalse 
\begin{figure}[!t]
\centering
\includegraphics[width=0.99\linewidth]{hardware.jpg}
\caption{Hardware test environment}\label{environment}
\end{figure} 
\fi

\begin{figure}[!t]
\begin{minipage}[t]{0.48\linewidth}
	\centering
	\includegraphics[width=.69\linewidth]{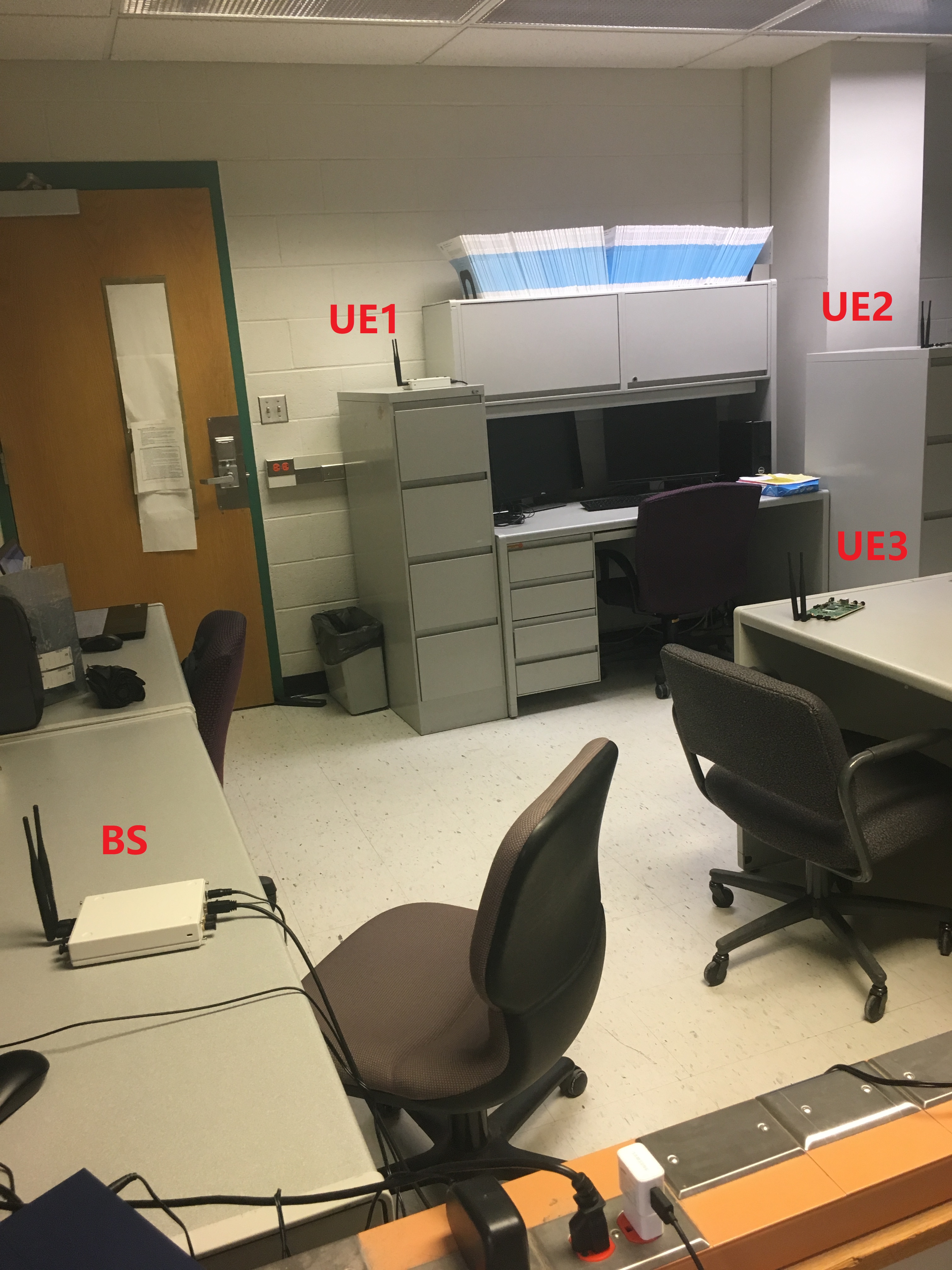}
	\caption{Hardware test environment}\label{environment}
\end{minipage}
\begin{minipage}[t]{0.48\linewidth}
	\centering
	\includegraphics[width=2.0in]{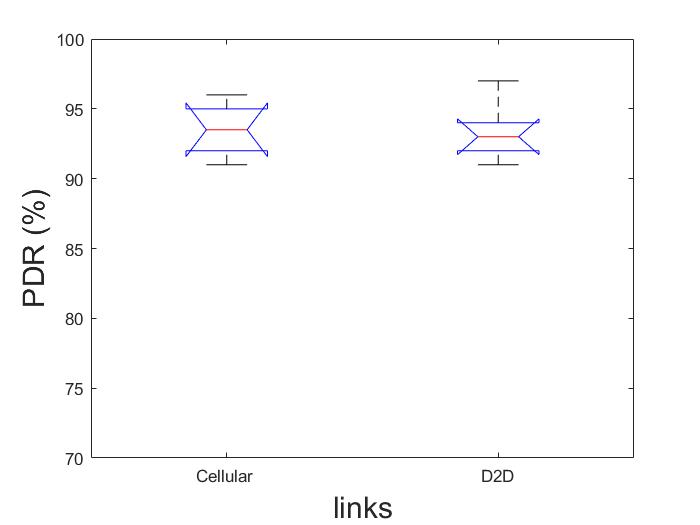}
	\caption{Reliability guarantee in the single cell LTE network}\label{measurement}
\end{minipage}
\end{figure}

\subHeadingS{Measurement results}
With the network setup mentioned above, the UE-to-UE communication pair chooses to operate in the D2D mode, thus forming a D2D link in the network. Figure~\ref{measurement} shows the packet-delivery-reliability (PDR) across the cellular link and D2D link. We see that both links have a PDR above the required reliability of 90\%. 
	In the prototyping cellular network, the deployment of the cellular link and D2D link is such that they interfere with each other, and UCS is able to discover the interference relation and schedule the two links so that they do not use the same carrier at the same time. Due to the availability of multiple carriers, however, the two links are able to transmit concurrently with each other at different channels, with channel hopping exercised along the way.
(Besides the communication reliability requirement of 90\%, we have experimented with other reliability requirements such as 80\% and 70\%, and we have observed similar system behavior.)

\iffalse 
\begin{figure}[!t]
\centering
\includegraphics[width=2.0in]{measurement.jpg}
\caption{Reliability guarantee in the single cell LTE network}\label{measurement}
\end{figure} 
\fi

\subsection{Simulation Evaluation}
Having validated the feasibility of realizing the UCS framework using commodity software and hardware platforms, we use the simulator of OpenAirInterface to evaluate the performance of UCS with at-scale, high-fidelity simulation. 

\subHeading{Simulation platform}
OpenAirInterface comes with a high-fidelity simulator for OpenAirLTE networks. The OpenAirInterface simulator enables simulation with the full PHY layer and synthetic radio channels, or with a PHY layer abstraction. In both PHY layer simulation modes, the full protocol stack of the UCS-variant of LTE is executed as is the case with USRP B210 hardware implementation and validation. The PHY layer of the OpenAirInterface provides the abstraction for all the transmission channels in LTE and the simulation for a large amount of channel types (e.g., Rayleigh, Rice and AWGN) with different parameter settings, which can cover a variety of actual scenarios. Consequently, the simulation in OpenAirInterface has a high degree of fidelity.

\subHeading{Simulation Scenarios} 
We focus on multi-cell scenarios where a total of 135 UEs are distributed in 9 cells, which are organized in a $3\times 3$ grid manner such that each cell covers a square area of $500m \times 500m$ and the 9 cells covers a square area of $1,500m \times 1,500m$. Each cell has 15 UEs deployed; five UEs are randomly chosen to communicate with the BS directly, forming five cellular links, with both uplink and downlink data transmissions; the remaining 10 UEs form five UE-to-UE communication pairs, with the source UE and destination UE randomly picked. Each UE-to-UE communication pair can be in cellular mode or D2D mode, with the specific mode to be selected by UCS. 

To understand the behavior of UCS in different network settings, we experiment with both random and grid network topologies, where the UEs of each cell are spatially distributed in a uniform-random and grid manner respectively. The base station (BS) of each cell is located at the center of the cell. We also experiment with different LTE  channel models of different path-loss exponents and fading models. To study UCS' capability of ensuring predictable control of interference and thus predictable communication reliability, we experiment with different packet-delivery-reliability (PDR) requirements of 80\%, 85\%, 90\%, and 95\%. 

The main parameters used in a typical simulation are summarized in Table~\ref{table:simulationParameters}. To understand the behavior of UCS in heterogeneous network settings, we also study the scenario where each UE randomly chooses a PDR requirement of 80\%, 85\%, 90\%, or 95\%, as well as the scenario where each transmitter randomly chooses a transmission power of 15dBm, 20dBm, or 25dBm.

\begin{table}[!htbp]
	\vspace*{-0.05in}
\caption{Simulation parameters used in OpenAirInterface simulation}
\label{table:simulationParameters}
\centering
\begin{tabular}{|c|c|}
\hline
\textbf{Parameter}&\textbf{Value}\\
\hline
\# of eNBs (i.e., base stations)&9\\
\hline
\# of UEs&135\\
\hline
Area of each cell&$500m\times 500m$\\
\hline
Topology&random, grid\\
\hline
\# of available carriers&25, 50, 100\\
\hline
Path-loss exponent&3.0, 3.5, 4.0\\
\hline
Fading type&Rayleigh, Rice, AWGN\\
\hline
eNB transmission power & 40dBm \\
\hline
UE transmission power & 20dBm  \\ 
\hline
Frequency band&band 7 (center freq.: 2.6GHz)\\
\hline
Simulation duration&10000 TTIs\\
\hline
PHY layer abstraction&Yes\\
\hline
Traffic type&Full buffer\\
\hline
Mobility&Static\\
\hline
\# of runs per expt$.$ config$.$&10\\
\hline
\end{tabular}
\vspace*{-0.1in}
\end{table}

\subHeading{Scheduling Protocols}
To understand the effectiveness of the UCS scheduling framework in ensuring predictable communication reliability and high channel spatial reuse, we compare it with the following scheduling protocols:
\begin{itemize}
\item \emph IAS: a interference-aware scheduling scheme that exploits the multi-user diversity of the cellular network such that the performance of the D2D underlay is optimized while maintaining a target performance level of the cellular network. The D2D terminals sense the radio spectrum and aid the BS in generating local awareness of the radio environment. The BS then uses this information in interference-aware resource allocation among the cellular and D2D links \cite{janis2009interference}.
\item \emph QAS: a QoS-aware scheduling scheme that utilizes channel statistical characteristics to maximize the overall throughput of the cellular users and admissible D2D pairs while guaranteeing a target signal-to-interference-plus-noise ratio (SINR) for each receiver. A D2D receiver only feeds back the channel-state-information (CSI) for a few best potential partner cellular users to reduce feedback overhead \cite{feng2016qos}.
\end{itemize} 

For understanding the optimality of the UCS scheduling framework, we also compare it with \emph{iOrder} \cite{che2014case}, a state-of-the-art centralized scheduling scheme that maximizes channel spatial reuse while ensuring the required PDRs by considering both the interference budget (i.e., tolerable interference power at receivers) and queue length in scheduling. When constructing the schedule for a time slot, iOrder first picks a 
link with the maximum number of queued packets; then iOrder adds links to the slot one at a time in a way that maximizes the interference budget at each step; this process repeats until no additional link can be added to the slot in any communication channel without violating the application requirement on link reliability \cite{che2014case}. When experimenting with iOrder, we assume that all the channel state information is available (which is unrealistic but serves as a reference for understanding the optimality of UCS), and we use the same set of cellular links and D2D links as those in the experiments for UCS. % (i.e., we assume the same mode-selection result as that in UCS-based experiments). 
(Note: the original iOrder algorithm was designed for single-channel settings. We extend it to multi-channel settings in this study by using the core idea of the iOrder algorithm in centrally scheduling concurrent transmissions for each channel.)

\begin{figure*}
	\begin{minipage}[t]{0.23\linewidth}
		\centering
		\includegraphics[width=0.99\linewidth]{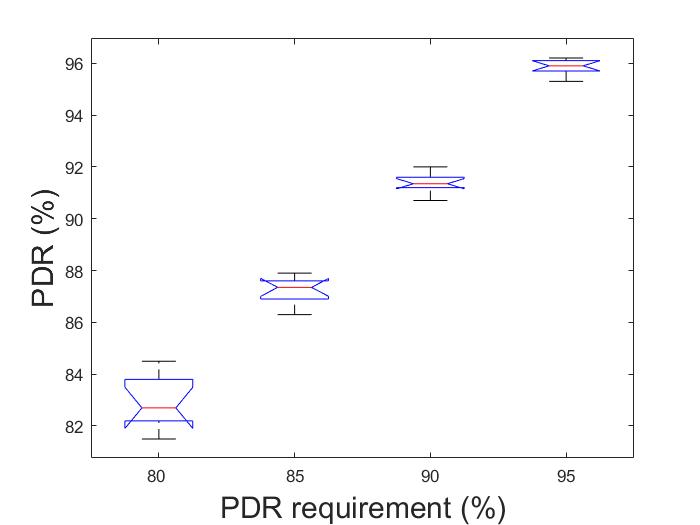}
		\caption{Same $K$ for adjacent 25 carriers}\label{25}
	\end{minipage}  
	\hfill
	\begin{minipage}[t]{0.23\linewidth}
		\centering
		\includegraphics[width=0.99\linewidth]{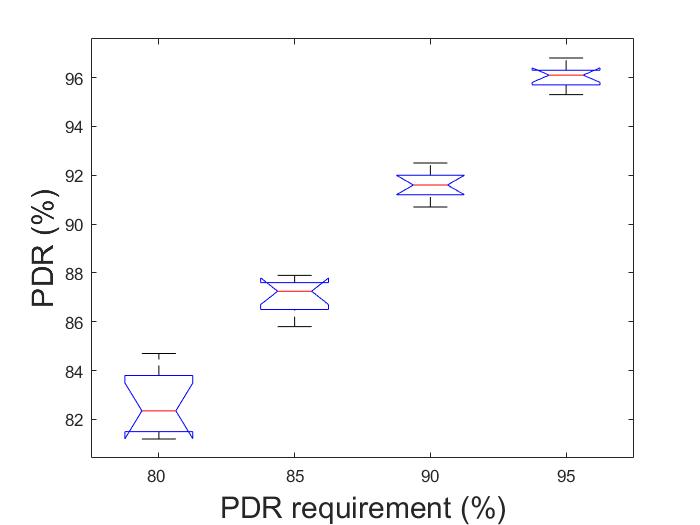}
		\caption{Same $K$ for adjacent 50 carriers}\label{50}
	\end{minipage}  
	\hfill
	\begin{minipage}[t]{0.23\linewidth}
		\centering
		\includegraphics[width=0.99\linewidth]{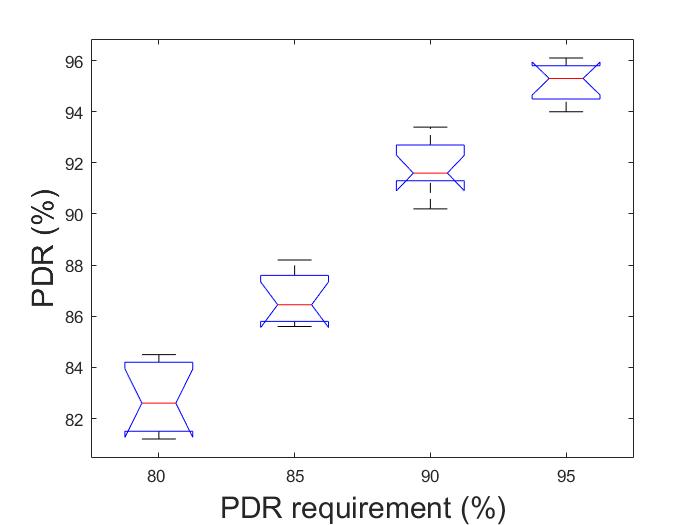}
		\caption{Same $K$ for adjacent 100 carriers}\label{100}
	\end{minipage} 
	\hfill
	\begin{minipage}[t]{0.23\linewidth}
		\centering
		\includegraphics[width=0.99\linewidth]{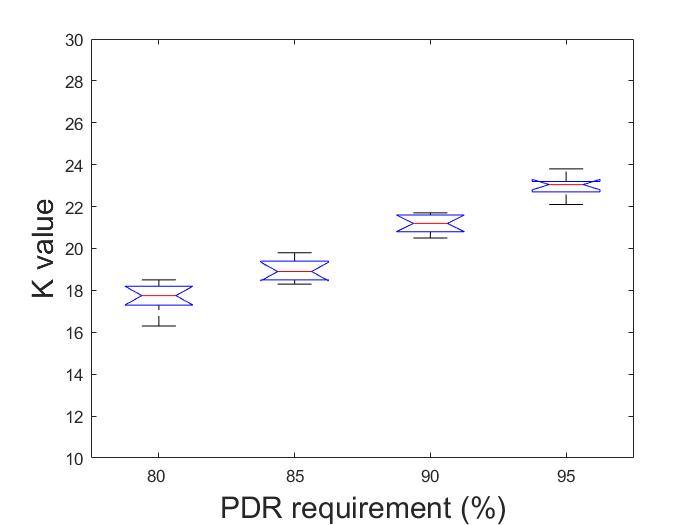}
		\caption{The value of the parameter $K$}\label{Kvalue}
	\end{minipage} 
	%\end{figure*}
	%\begin{figure*}
	
	%\\ 	%\vspace*{0.1in}
	\begin{minipage}[t]{0.23\linewidth}
		\centering
		\includegraphics[width=0.99\linewidth]{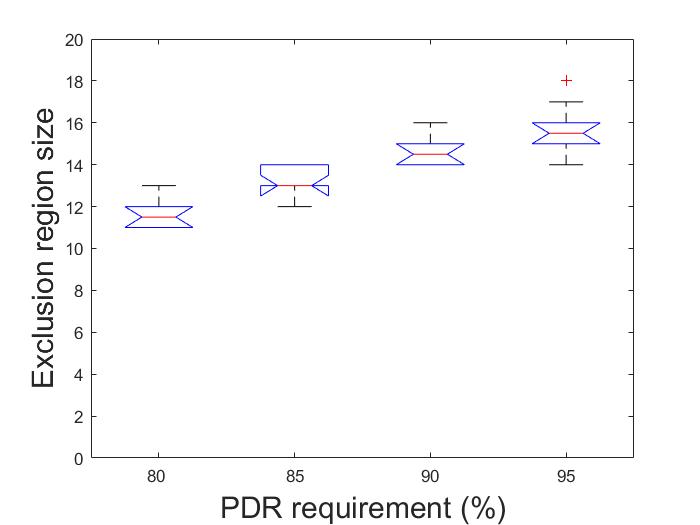}
		\caption{Exclusion region size}\label{ERsize}
	\end{minipage}
	\hfill
	\begin{minipage}[t]{0.23\linewidth}
		\centering
		\includegraphics[width=0.99\linewidth]{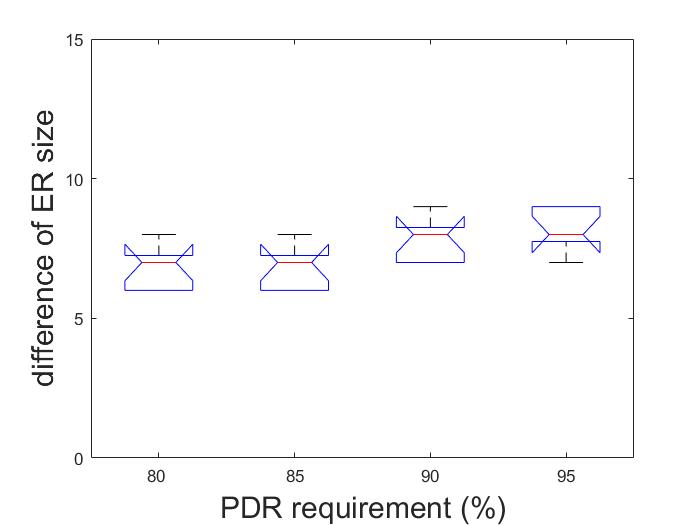}
		\caption{Comparison of ER size in D2D and cellular modes for UE-to-UE pairs choosing D2D mode}\label{difference}
	\end{minipage} 
	\hfill
	\begin{minipage}[t]{0.23\linewidth}
		\centering
		\includegraphics[width=0.99\linewidth]{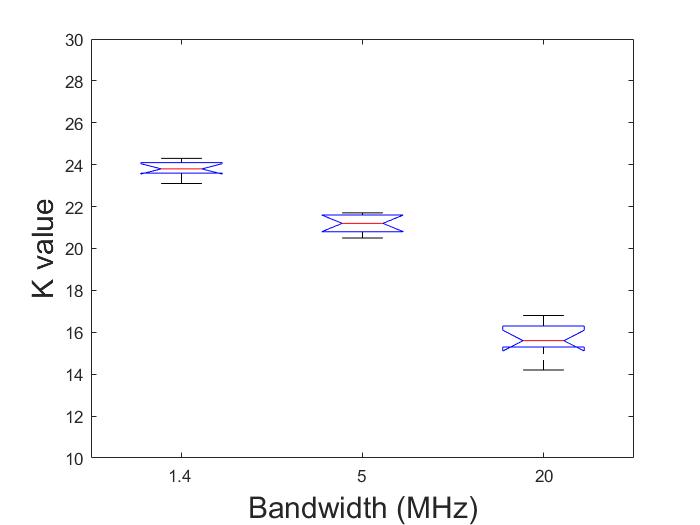}
		\caption{The value of the parameter $K$ for different bandwidth}\label{Kvalueband}
	\end{minipage}
	\hfill
	\begin{minipage}[t]{0.23\linewidth}
		\centering
		\includegraphics[width=0.99\linewidth]{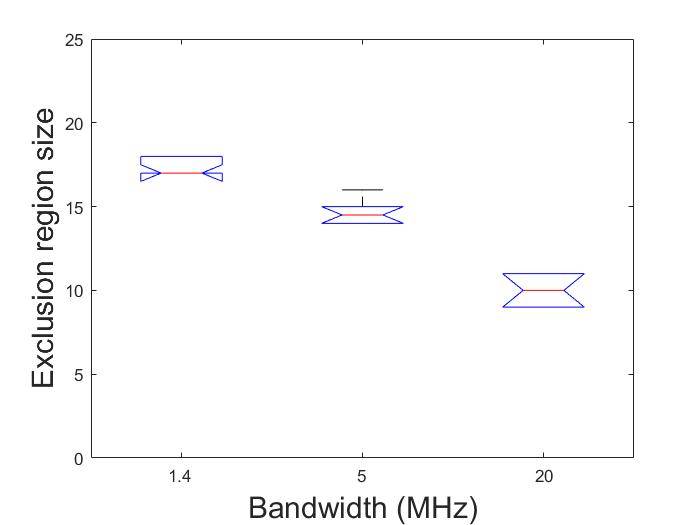}
		\caption{Exclusion region size for different bandwidth}\label{ERsizeband}
	\end{minipage} 
	\hfill
    \begin{minipage}[t]{0.23\linewidth}
	\centering
	\includegraphics[width=.99\linewidth]{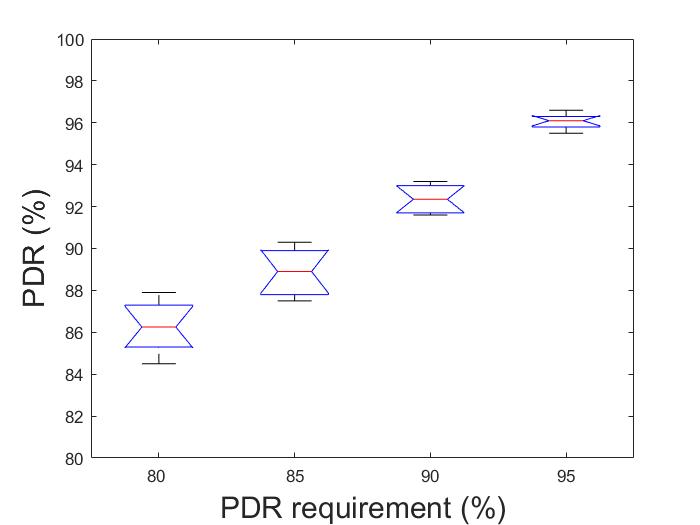}
	\caption{Communication reliability in heterogeneous reliability scenario}\label{mixreliability}
    \end{minipage}
    \hfill
    \begin{minipage}[t]{0.23\linewidth}
	\centering
	\includegraphics[width=.99\linewidth]{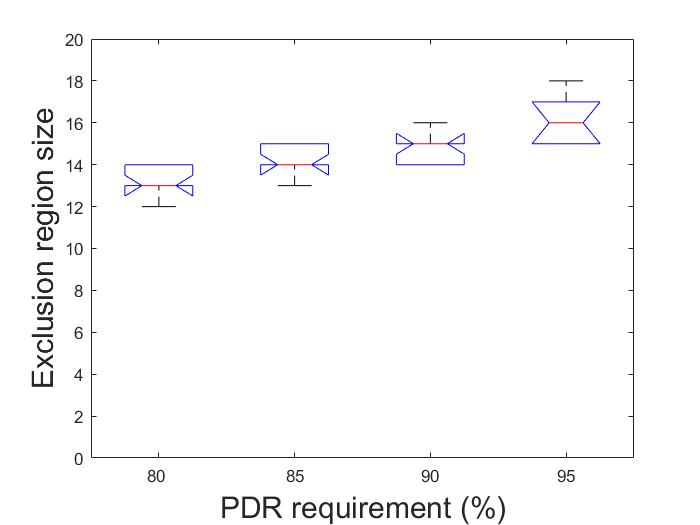}
	\caption{Exclusion region size in heterogeneous reliability scenario}\label{ERmixedPDR}
    \end{minipage}
	\hfill
	\begin{minipage}[t]{0.23\linewidth}
	\centering
	\includegraphics[width=.99\linewidth]{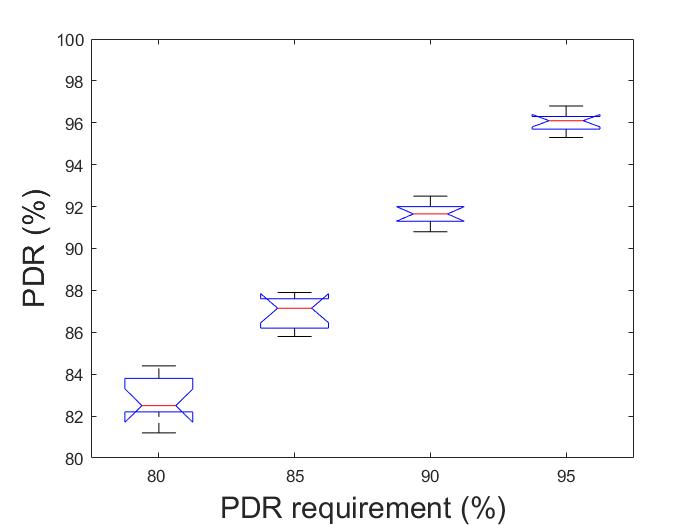}
	\caption{Communication reliability in heterogeneous power scenario}\label{mixpower}
	\end{minipage}
	\hfill
	\begin{minipage}[t]{0.23\linewidth}
	\centering
	\includegraphics[width=.99\linewidth]{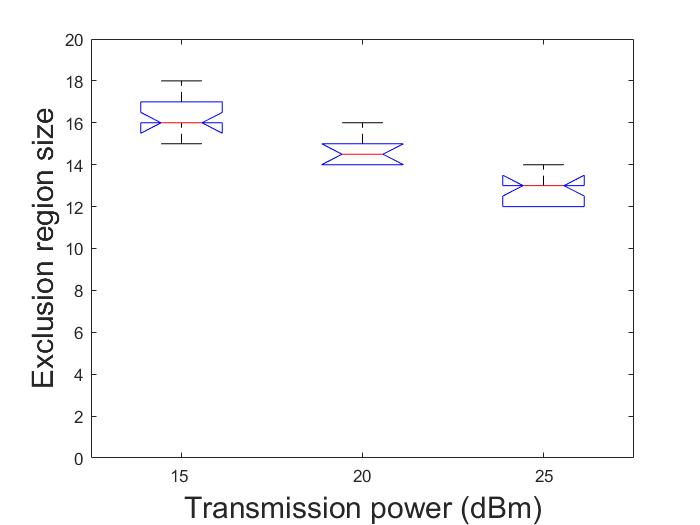}
	\caption{Exclusion region size in heterogeneous power scenario}\label{ERmixedpower}
	\end{minipage}
	\hfill
\\
	\begin{minipage}[t]{0.23\linewidth}
	\centering
	\includegraphics[width=.99\linewidth]{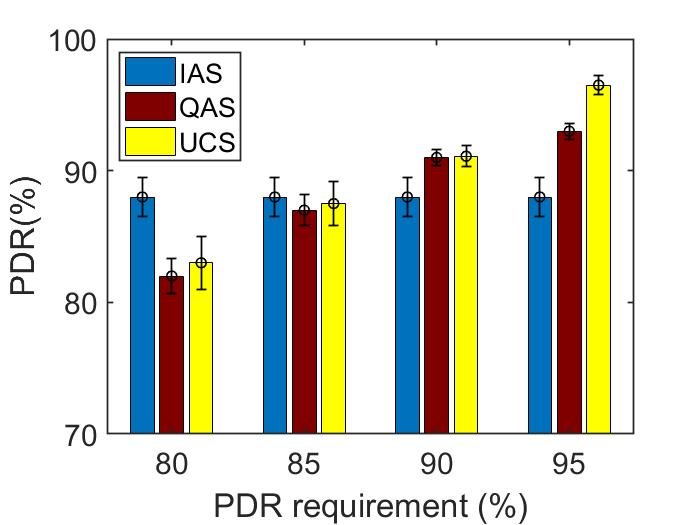}
	\caption{Predictable reliability guarantee for random topology}\label{reliability}
\end{minipage}
\hfill
\begin{minipage}[t]{0.23\linewidth}
	\centering
	\includegraphics[width=.99\linewidth]{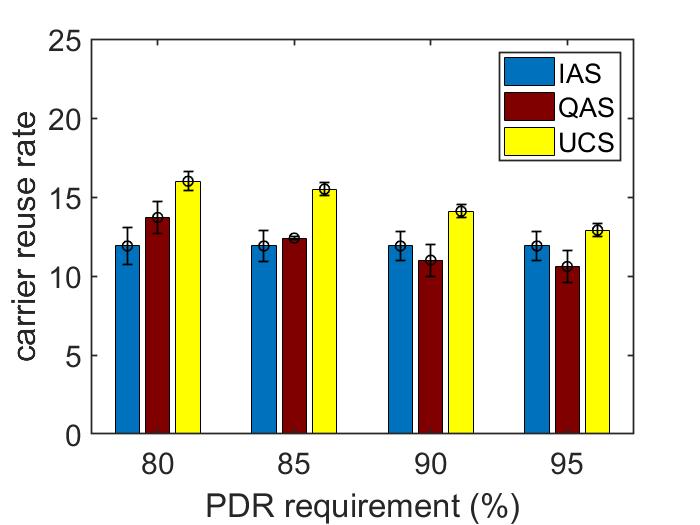}
	\caption{Mean carrier reuse rate with the random topology}\label{randomreuse}
\end{minipage}
\hfill
\begin{minipage}[t]{0.23\linewidth}
	\centering
	\includegraphics[width=.99\linewidth]{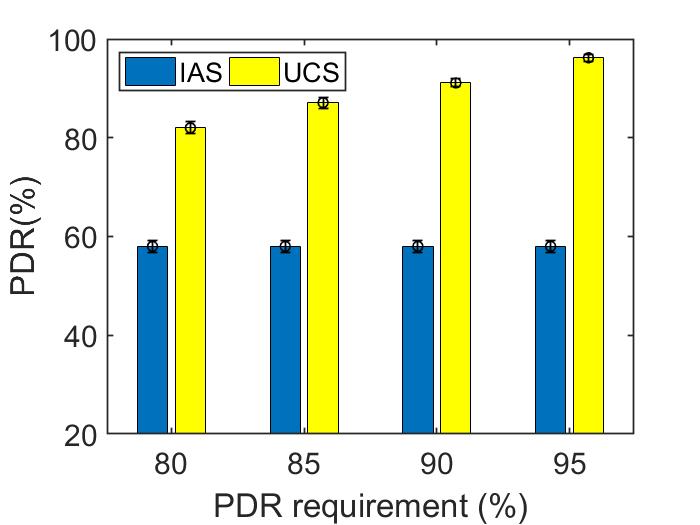}
	\caption{Predictable reliability guarantee comparison with IAS}\label{reliabilityrange}
\end{minipage}
\hfill
\begin{minipage}[t]{0.23\linewidth}
	\centering
	\includegraphics[width=.99\linewidth]{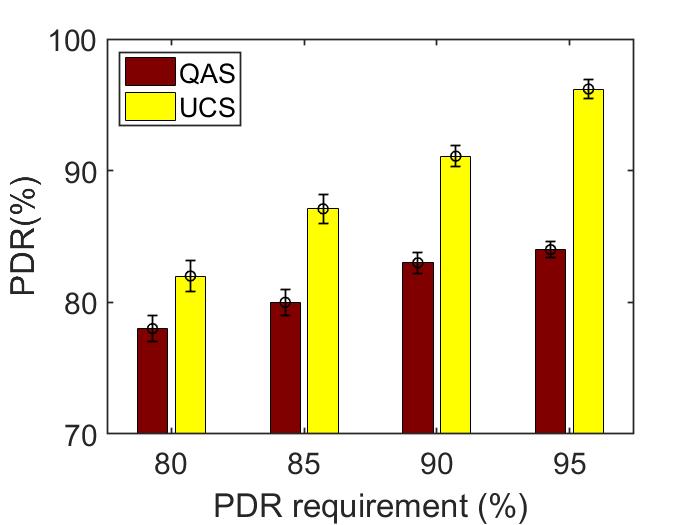}
	\caption{Predictable reliability guarantee comparison with QAS}\label{reliabilitychannel}
\end{minipage}
\vspace*{-0.1in}
\end{figure*}

\begin{figure}[!t]
\begin{minipage}[t]{0.48\linewidth}
	\centering
	\includegraphics[width=.99\linewidth]{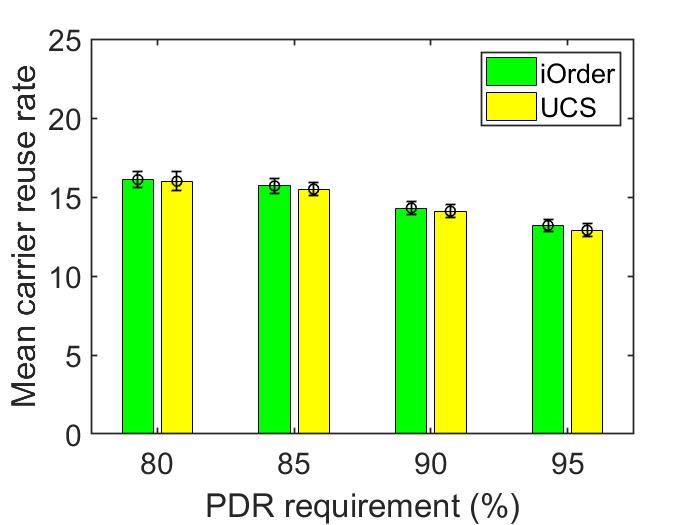}
	\caption{Comparison with iOrder under the random topology}\label{randomreusemax}
\end{minipage}
\begin{minipage}[t]{0.48\linewidth}
	\centering
	\includegraphics[width=.99\linewidth]{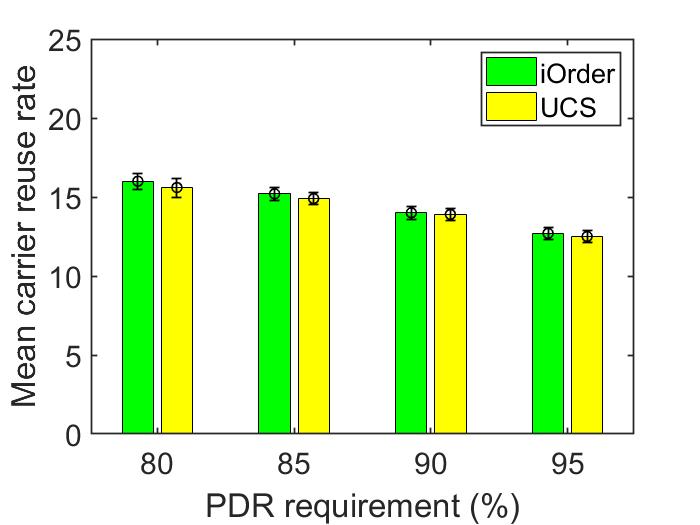}
	\caption{Comparison with iOrder under the grid topology}\label{gridreusemax}
\end{minipage}
\end{figure}

\iffalse 
\begin{figure*}
%\vspace*{-0.05in}
	\begin{minipage}[t]{0.23\linewidth}
		\centering
		\includegraphics[width=.99\linewidth]{reliability.jpg}
		\caption{Predictable reliability guarantee for random topology}\label{reliability}
	\end{minipage}
	\hfill
	\begin{minipage}[t]{0.23\linewidth}
		\centering
		\includegraphics[width=.99\linewidth]{randomreuse.jpg}
		\caption{Mean carrier reuse rate with the random topology}\label{randomreuse}
	\end{minipage}
	\hfill
	\begin{minipage}[t]{0.23\linewidth}
		\centering
		\includegraphics[width=.99\linewidth]{reliabilityrange.jpg}
		\caption{Predictable reliability guarantee comparison with IAS}\label{reliabilityrange}
	\end{minipage}
	\\
	\begin{minipage}[t]{0.23\linewidth}
		\centering
		\includegraphics[width=.99\linewidth]{reliabilitychannel.jpg}
		\caption{Predictable reliability guarantee comparison with QAS}\label{reliabilitychannel}
	\end{minipage}
	\hfill
	\begin{minipage}[t]{0.23\linewidth}
		\centering
		\includegraphics[width=2.0in]{reuserandommax.jpg}
		\caption{Comparison with iOrder under the random topology}\label{randomreusemax}
	\end{minipage}
	\hfill
	\begin{minipage}[t]{0.23\linewidth}
		\centering
		\includegraphics[width=2.0in]{reusegridmax.jpg}
		\caption{Comparison with iOrder under the grid topology}\label{gridreusemax}
	\end{minipage}
	\vspace*{-0.1in}
\end{figure*}
\fi

\subHeading{Simulation results}
To understand UCS' capability in ensuring predictable communication reliability, we measure the actual packet-delivery-reliability (PDR) when using a same PRK model parameter $K$ for different number of carriers (i.e., 25, 50 and 100). From Figures~\ref{25} and ~\ref{50}, we see that maintaining one $K$ for 25 and 50 adjacent carriers can ensure application required PDRs. Figure~\ref{100} shows that maintaining one $K$ for 100 adjacent carriers also can also guarantee PDRs up to 90\%, but it does not always ensure the required high PDR of 95\%. We see that, by choosing the right number $n$ of carriers/channels for which to maintain a single PRK model parameter $K$, UCS can ensure the required PDRs. The fact that too large a $n$ (i.e., 100 in this study) cannot always ensure the required high PDRs also demonstrate the tradeoff between control signaling overhead, modeling accuracy, and protocol performance, as discussed in Section~\ref{PRK adaptation}. 
For the figures in the rest of this section, we use the data for scenarios when UCS maintains a PRK model parameter $K$ for every 25 adjacent carriers. 

The reason why UCS ensures the required communication reliability is because it adapts the PRK model parameter and the thus the size of the exclusion region (ER) around each receiver (i.e., number of nodes in the ER) according to the application-required PDR and in-situ network and environmental conditions. As shown by Figures~\ref{Kvalue} and~\ref{ERsize} for cellular links and D2D links respectively, higher PDRs are achieved by increasing the value of parameter $K$ and thus expanding the ER size.

To understand the frequency reuse efficiency for the mode selection algorithm of the UCS framework, Figure~\ref{difference} shows, for the set of UE-to-UE communication pairs operating in the D2D mode, the ER size in the cellular mode minus that in the D2D mode, with the ER size for the cellular mode calculated as the sum of the ER sizes of the involved uplink and downlink. We see that, for these communication pairs, operating in the D2D mode significantly reduces the ER size and thus improves the transmission concurrency and channel spatial reuse. 

To understand the effect of the number of available carriers on the behavior of UCS, we compared the value of parameter $K$ and the size of exclusion region when different communication bandwidth is used. For instance, Figure~\ref{Kvalueband} and~\ref{ERsizeband} shows the values of K and ER size when the communication PDR requirement is 90\%. With the increase of available carriers, the interference between concurrent transmitting links decreases. Accordingly, UCS reduces the size of exclusion regions by adjusting the value of parameter $K$ dynamically to ensure the required communication reliability  while maximizing carrier spatial reuse.

For the heterogenous network setting where each link randomly chooses a PDR requirement, Figures~\ref{mixreliability}  and~\ref{ERmixedPDR} show the communication reliability and corresponding ER size for each reliability requirement. We see that UCS ensures application-required PDR in these settings too. Compared with the homogeneous scenario where each link has the same PDR requirement, the average communication PDR of some links with lower reliability requirements (e.g., 80\%) is a little higher. This is because the heterogeneous setting has more links with higher reliability requirements, which tend to have larger ERs as reflected in Figure~\ref{ERmixedPDR}.

For the heterogenous network setting where each transmitter randomly picks a transmission power, Figure~\ref{mixpower} shows the communication reliability in UCS. We see that UCS ensures application-required PDR. For the communication pairs with PDR requirement of 90\%, Figure~\ref{ERmixedpower} shows the ER size for links with different transmission powers. We see that transmission power can influence the ER size, and thus affecting communication concurrency and throughput. The links with the lower transmission power tend to tolerate lower interference power, thus they maintain larger ERs to guarantee the required communication reliability.

Figures~\ref{reliability} and~\ref{randomreuse} show the PDR and carrier reuse rate in different protocols respectively. We see that the existing cellular protocols IAS and QAS cannot ensure predictable interference control (as explained in Section~\ref{sec:introduction}) and thus cannot ensure predictable communication reliability, for instance, not able to ensure the required high reliability of 95\%. 

To highlight the advantages of UCS in the predictable communication reliability guarantee, we further compared UCS with IAS and QAS separately in Figure~\ref{reliabilityrange} and Figure~\ref{reliabilitychannel} in different network setting. In Figure~\ref{reliabilityrange}, we reduced the area of each cell to $150m*150m$ to increase the interference of concurrent transmission links; IAS cannot even guarantee the communication reliability up to 60\% because it only avoids the strong interference, not considering the interference accumulation.  
In Figure~\ref{reliabilitychannel}, we changed the channel type from Rayleigh model to Rice model; QAS cannot guarantee the communication reliability up to 85\% because it only targets for Rayleigh model, not considering potentially different channel models and reality. In contrast, UCS ensures the required PDR for different cell ranges and channel models while achieving a higher carrier reuse rate (which is defined as the number of links using a carrier at each time slot). In fact, Figures~\ref{randomreusemax} and \ref{gridreusemax} show the distributed UCS scheduling framework achieves a carrier reuse rate statistically equal to that of the centralized, state-of-the-art scheduling algorithm iOrder, showing the optimality of the UCS framework \footnote{The carrier reuse rates in the random and grid networks are statistically similar.}.

\section{Conclusion}\label{sec:conclusion}
We have proposed a field-deployable, unified cellular scheduling framework UCS to ensure predictable communication reliability in cellular networks with D2D links. The UCS framework effectively leverages the PRK interference model and addresses the challenges of multi-channel PRK-based scheduling in cellular networks. UCS provides a simple mode selection mechanism which, together with PRK-based cellular scheduling, maximizes communication throughput while ensuring communication reliability. UCS also effectively leverages cellular network structures (e.g., the availability of BSes and high-speed, out-of-band networks in-between) to orchestrate BS and UE functionalities for light-weight control signaling and ease of incremental deployment and technology evolution. The feasibility and performance of UCS have been verified through high-fidelity simulation and real-world hardware and software implementation.

{\small
\bibliography{ref}

\begin{thebibliography}{10}
\providecommand{\url}[1]{#1}
\csname url@samestyle\endcsname
\providecommand{\newblock}{\relax}
\providecommand{\bibinfo}[2]{#2}
\providecommand{\BIBentrySTDinterwordspacing}{\spaceskip=0pt\relax}
\providecommand{\BIBentryALTinterwordstretchfactor}{4}
\providecommand{\BIBentryALTinterwordspacing}{\spaceskip=\fontdimen2\font plus
\BIBentryALTinterwordstretchfactor\fontdimen3\font minus
  \fontdimen4\font\relax}
\providecommand{\BIBforeignlanguage}[2]{{%
\expandafter\ifx\csname l@#1\endcsname\relax
\typeout{** WARNING: IEEEtran.bst: No hyphenation pattern has been}%
\typeout{** loaded for the language `#1'. Using the pattern for}%
\typeout{** the default language instead.}%
\else
\language=\csname l@#1\endcsname
\fi
#2}}
\providecommand{\BIBdecl}{\relax}
\BIBdecl

\bibitem{Erik-LTE-book:2016}
E.~Dahlman, S.~Parkvall, and J.~Skold, \emph{{4G, LTE-Advanced Pro and The Road
  to 5G}}.\hskip 1em plus 0.5em minus 0.4em\relax Academic Press, 2016.

\bibitem{5G-MTC:overview}
S.~Essakiappan, S.~Harb, and A.~Solar-schultz, ``{Machine-Type Communications:
  Current Status and Future Perspectives Toward \textsc{5G} Systems},''
  \emph{IEEE Communications Magazine}, September 2015.

\bibitem{5G-D2D:overview}
M.~N. Tehrani, M.~Uysal, and H.~Yanikomeroglu, ``{Device-to-device
  communication in \textsc{5G} cellular networks: challenges, solutions, and
  future directions},'' \emph{IEEE Communications Magazine}, may 2014.

\bibitem{PRKS}
H.~Zhang, X.~Liu, C.~Li, Y.~Chen, X.~Che, L.~Y. Wang, F.~Lin, and G.~Yin,
  ``Scheduling with predictable link reliability for wireless networked
  control,'' \emph{IEEE Transactions on Wireless Communications}, vol.~16,
  no.~9, 2017.

\bibitem{PRK}
H.~Zhang, X.~Che, X.~Liu, and X.~Ju, ``Adaptive instantiation of the protocol
  interference model in wireless networked sensing and control,'' \emph{ACM
  Transactions on Sensor Networks (TOSN)}, vol.~10, no.~2, 2014.

\bibitem{control-comm-netRT}
J.~R. Moyne and D.~M. Tilbury, ``Control and communication challenges in
  networked real-time systems,'' \emph{Proceedings of the IEEE}, vol.~95,
  no.~1, 2007.

\bibitem{Hyunkee:capacity-d2d}
H.~Min, J.~Lee, S.~Park, and D.~Hong, ``Capacity enhancement using an
  interference limited area for \textsc{D}evice-to-\textsc{D}evice uplink
  underlaying cellular networks,'' \emph{IEEE Transactions on Wireless
  Communications}, vol.~10, no.~12, 2011.

\bibitem{Qiaoyang:resource-d2d}
Q.~Ye, M.~Al-Shalash, C.~Caramanis, and J.~G. Andrews, ``Distributed resource
  allocation in \textsc{D}evice-to-\textsc{D}evice enhanced cellular
  networks,'' \emph{IEEE Transactions on Communications}, vol.~63, no.~2, 2015.

\bibitem{Jeff:D2D-resource-allication-distributed}
Q.~Ye, M.~Al-shalash, C.~Caramanis, and J.~A. and, ``{Distributed Resource
  Allocation in Device-to-Device Enhanced Cellular Networks},'' \emph{IEEE
  Transactions on Communications}, vol.~63, no.~2, 2015.

\bibitem{Xuemin:resource-control-d2d}
L.~Lei, Y.~Kuang, X.~Shen, C.~Lin, and Z.~Zhong, ``Resource control in network
  assisted \textsc{D}evice-to-\textsc{D}evice communications: Solutions and
  challenges,'' \emph{IEEE Communications Magazine}, vol.~52, no.~6, 2014.

\bibitem{Vincent:dynamic-power}
W.~Wang, F.~Zhang, and V.~K.~N. Lau, ``Dynamic power control for delay-aware
  \textsc{D}evice-to-\textsc{D}evice communications,'' \emph{IEEE Journal on
  Selected Areas in Communications}, vol.~33, no.~1, 2015.

\bibitem{Vincent:delay-aware-control-survey}
Y.~Cui, V.~K.~N. Lau, R.~Wang, H.~Huang, and S.~Zhang, ``{A Survey on
  Delay-Aware Resource Control for Wireless Systems: Large Deviation Theory,
  Stochastic Lyapunov Drift, and Distributed Stochastic Learning},'' \emph{IEEE
  Transactions on Information Theory}, vol.~58, no.~3, 2012.

\bibitem{D2D:Verenzuela2017}
D.~Verenzuela and G.~Miao, ``{Scalable D2D Communications for Frequency Reuse
  {\textgreater}{\textgreater}1 in 5G},'' vol.~16, no.~6, 2017.

\bibitem{D2D:Lv2017}
S.~Lv, C.~Xing, Z.~Zhang, and K.~Long, ``{Guard Zone Based Interference
  Management for D2D-Aided Underlaying Cellular Networks},'' \emph{IEEE
  Transactions on Vehicular Technology}, vol.~66, no.~6, 2017.

\bibitem{doppler2009device}
K.~Doppler, M.~Rinne, C.~Wijting, C.~B. Ribeiro, and K.~Hugl,
  ``Device-to-device communication as an underlay to lte-advanced networks,''
  \emph{IEEE Communications Magazine}, vol.~47, no.~12, 2009.

\bibitem{peng2009interference}
T.~Peng, Q.~Lu, H.~Wang, S.~Xu, and W.~Wang, ``Interference avoidance
  mechanisms in the hybrid cellular and device-to-device systems,'' in
  \emph{Personal, Indoor and Mobile Radio Communications, 2009 IEEE 20th
  International Symposium on}.\hskip 1em plus 0.5em minus 0.4em\relax IEEE,
  2009.

\bibitem{min2011capacity}
H.~Min, J.~Lee, S.~Park, and D.~Hong, ``Capacity enhancement using an
  interference limited area for device-to-device uplink underlaying cellular
  networks,'' \emph{IEEE Transactions on Wireless Communications}, vol.~10,
  no.~12, 2011.

\bibitem{yu2009performance}
C.-H. Yu, K.~Doppler, C.~Ribeiro, and O.~Tirkkonen, ``Performance impact of
  fading interference to device-to-device communication underlaying cellular
  networks,'' in \emph{Personal, Indoor and Mobile Radio Communications, 2009
  IEEE 20th International Symposium on}.\hskip 1em plus 0.5em minus 0.4em\relax
  IEEE, 2009.

\bibitem{janis2009interference}
P.~Janis, V.~Koivunen, C.~Ribeiro, J.~Korhonen, K.~Doppler, and K.~Hugl,
  ``Interference-aware resource allocation for device-to-device radio
  underlaying cellular networks,'' in \emph{Vehicular Technology Conference,
  2009. VTC Spring 2009. IEEE 69th}.\hskip 1em plus 0.5em minus 0.4em\relax
  IEEE, 2009.

\bibitem{feng2016qos}
D.~Feng, L.~Lu, Y.-W. Yi, G.~Y. Li, G.~Feng, and S.~Li, ``Qos-aware resource
  allocation for device-to-device communications with channel uncertainty,''
  \emph{IEEE Transactions on Vehicular Technology}, vol.~65, no.~8, 2016.

\bibitem{4GLTE}
E.~Dahlman, S.~Parkvall, and J.~Skold, \emph{4G LTE/LTE-Advanced for Mobile
  Broadband}.\hskip 1em plus 0.5em minus 0.4em\relax Elsevier Ltd, 2011.

\bibitem{min2011reliability}
H.~Min, W.~Seo, J.~Lee, S.~Park, and D.~Hong, ``Reliability improvement using
  receive mode selection in the device-to-device uplink period underlaying
  cellular networks,'' \emph{IEEE Transactions on Wireless Communications},
  vol.~10, no.~2, 2011.

\bibitem{IndustrialNetBook-CRC}
R.~Z. (Editor), \emph{Industrial Communication Technology Handbook}.\hskip 1em
  plus 0.5em minus 0.4em\relax CRC Press, 2015.

\bibitem{Lu:WirelessHART-Sch}
A.~Saifullah, Y.~Xu, C.~Lu, and Y.~Chen, ``Real-time scheduling in
  \textsc{W}ireless\textsc{HART} networks,'' in \emph{IEEE RTSS}, 2010.

\bibitem{che2014case}
X.~Che, H.~Zhang, and X.~Ju, ``The case for addressing the ordering effect in
  interference-limited wireless scheduling,'' \emph{IEEE Transactions on
  Wireless Communications}, vol.~13, no.~9, 2014.

\bibitem{CPS-IOTDI18}
C.~Li, H.~Zhang, J.~Rao, L.~Y. Wang, and G.~Yin, ``{Cyber-Physical Scheduling
  for Predictable Reliability of Inter-Vehicle Communications},'' in
  \emph{ACM/IEEE IoTDI (short paper)}, 2018.

\bibitem{Zhang:pktReliability}
L.~Wang, H.~Zhang, and P.~Ren, ``Distributed scheduling and power control for
  predictable iot communication reliability,'' in \emph{IEEE ICC}, 2018.

\bibitem{tekin2015online}
C.~Tekin, M.~Liu \emph{et~al.}, ``Online learning methods for networking,''
  \emph{Foundations and Trends{\textregistered} in Networking}, vol.~8, no.~4,
  2015.

\bibitem{liu2015maximal}
X.~Liu, Y.~Chen, and H.~Zhang, ``A maximal concurrency and low latency
  distributed scheduling protocol for wireless sensor networks,''
  \emph{International Journal of Distributed Sensor Networks}, vol.~11, no.~8,
  2015.

\bibitem{OAI}
N.~Nikaein, M.~K. Marina, S.~Manickam, A.~Dawson, R.~Knopp, and C.~Bonnet,
  ``{OpenAirInterface: A Flexible Platform for 5G Research},'' \emph{ACM
  SIGCOMM Computer Communication Review}, vol.~44, no.~5, 2014.

\bibitem{virdis2015performance}
A.~Virdis, N.~Iardella, G.~Stea, and D.~Sabella, ``Performance analysis of
  openairinterface system emulation,'' in \emph{Future Internet of Things and
  Cloud (FiCloud), 2015 3rd International Conference on}.\hskip 1em plus 0.5em
  minus 0.4em\relax IEEE, 2015.

\end{thebibliography}
}
\end{document}